\documentclass[final,3p,times,numbers,sort&compress]{elsarticle}

\usepackage{amsfonts,amssymb,stmaryrd,dsfont}
\usepackage{amsbsy}
\usepackage{graphicx}
\usepackage{epstopdf}
\usepackage{galois}
\usepackage{latexsym}
\usepackage{amsmath,mathrsfs,stmaryrd}
\usepackage{psfrag}
\usepackage{subfigure,multicol}
\usepackage{url}
\makeatletter
  \newcommand\figcaption{\def\@captype{figure}\caption}
  \newcommand\tabcaption{\def\@captype{table}\caption}
\makeatother

\usepackage{amsthm}

\newtheorem{definition}{Definition}
\newtheorem{remark}{Remark}
\newtheorem{example}{Example}
\newtheorem{proposition}{Proposition}
\newtheorem{theorem}{Theorem}
\newtheorem{corollary}{Corollary}

\begin{document}

\begin{frontmatter}

\title{Making Problems Tractable on Big Data via Preprocessing\\ with Polylog-size Output}



\author[mymainaddress,mysecondaryaddress]{Jiannan Yang}
\ead{yjn19920627@pku.edu.cn}

\author[mymainaddress,mysecondaryaddress]{Hanpin Wang}
\ead{whpxhy@pku.edu.cn}

\author[mymainaddress,mysecondaryaddress]{Yongzhi Cao\corref{mycorrespondingauthor}}
\cortext[mycorrespondingauthor]{Corresponding author}
\ead{caoyz@pku.edu.cn}

\address[mymainaddress]{Institute of Software, School of Electronics Engineering and Computer Science,
Peking University, Beijing 100871, China}
\address[mysecondaryaddress]{Key Laboratory of High Confidence Software Technologies, Peking University, Ministry of Education, China}

\begin{abstract}
To provide a dichotomy between those queries that can be made feasible on big data after appropriate preprocessing and those for which preprocessing does not help, Fan et al. developed the $\sqcap$-tractability theory.
This theory provides a formal foundation for understanding the tractability of query classes in the context of big data. Along this line, we introduce a novel notion of $\sqcap'$-tractability in this paper.
Inspired by some technologies used to deal big data, we place a restriction on preprocessing function, which limits the function to produce a relatively small database as output, at most polylog-size of the input database.
At the same time, we bound the redundancy information when re-factorizing data and queries for preprocessing.
These changes aim to make our theory more closely linked to practice.
We set two complexity classes to denote the classes of Boolean queries that are $\sqcap'$-tractable themselves and that can be made $\sqcap'$-tractable, respectively.
Based on a new factorization in our complexity classes, we investigate two reductions, which differ from whether allowing re-factorizing data and query parts.
We verify the transitive and compatible properties of the reductions and analysis the complete problems and sizes of the complexity classes.
We conclude that all PTIME classes of Boolean queries can be made $\sqcap'$-tractable, similar to that of the $\sqcap$-tractability theory.
With a little surprise, we prove that the set of all $\sqcap'$-tractable queries is strictly smaller than that of all $\sqcap$-tractable queries, and thus the set of $\sqcap'$-tractable queries is properly contained in that of PTIME queries.
In this way, we attain a new complexity class inside the complexity class of PTIME queries.
\end{abstract}

\end{frontmatter}

\section{Introduction}
Nowadays, big data has become an increasingly hot topic in industry, academe, and government around the world \cite{nrc2013}.
While the rise of big data provides interesting new insights and great opportunities in many fields, it also presents new challenges for information technology at every step of big data analysis pipeline \cite{jagadish2014big}, due to the feature that the dataset size is beyond the ability of typical database software tools to capture, store, manage, and analyse (see, for example, \cite{michael2013big,philip2014data,sarma2013consistent,mozafari2013high}).

As a result, the traditional view on the tractability of query classes, among others, has been impacted in the context of big data \cite{fan2013making,wu2014data,jung2012synchronous}.
More concretely, a query class is traditionally considered tractable if there is an algorithm for answering the queries in time bounded by a polynomial (PTIME) in the size of database and query \cite{papadimitriou2003computational}.
However, in the background of big data, such an algorithm may become infeasible in practice. For instance, even with disk scanning speed of $6$GB$/$s, a linear scan of a dataset of $1$ PetaByte (PB, $10^{15}$ bytes) will take $1.9$ days \cite{fan2013making}.
This suggests that PTIME is a coarse metric for characterizing query classes under the background of big data.
To reduce the query time, some new technologies, including data preprocessing \cite{fan2013making,jung2012synchronous}, scale independence of queries \cite{fan2014scale,cao2014bounded}, query-preserving data compression \cite{fan2012query}, cloud computing \cite{hashem2015rise}, and parallel, distributed algorithm (MapReduce) \cite{dean2008mapreduce,kim2015supporting},  have been proposed \cite{fan2014querying,marz2015big}.
In particular, the preprocessing step that can be considered as an off-line process with a one-time cost has been adopted in some applications. After preprocessing a database of $1$ PB by building indices, one can answer all queries in seconds rather than $1.9$ days \cite{fan2013making}.

To specify the tractability of query classes in the context of big data with preprocessing, Fan et al. \cite{fan2013making} established a formal foundation in terms of computational complexity. In particular, the notion of $\sqcap$-tractability is proposed to describe the query classes that can be answered in parallel poly-logarithmic time (NC) after a PTIME preprocessing.
Moreover, they discussed the query class $\sqcap{}T_Q^0$, the set of all $\sqcap$-tractable queries, and showed that $\sqcap{}T_Q^0$ is properly contained in P, unless $\text{P}=\text{NC}$.
For the queries that are not $\sqcap$-tractable, they discussed the query classes $\sqcap{}T_Q$ that can be effectively converted to $\sqcap$-tractable queries by re-factorizing data and queries for preprocessing.
Correspondingly, the decision problems that can be made $\sqcap$-tractable, denoted by $\sqcap{}T_P$, have been considered as well.
Under NC-factor reducibility, the Breadth-Depth Search problem (BDS) is proven to be complete for $\sqcap{}T_P$.
Accordingly, the query version of BDS is $\sqcap{}T_Q$-complete. These fundamental results make a first step towards understanding the tractability of queries in the context of big data.

Remarkably, a closer examination of the preprocessing functions utilized shows that some of them have a restriction on the size of outputs.
For example, in \cite{fan2014queryingbiggraphs}, queries on a graph $G$ only access a fraction $G_Q$ of $G$ with size $|G_Q|\leq\alpha|G|$,
where $\alpha$ is a small ratio.
In \cite{cao2014bounded}, the small database accessed is bounded to contain at most constant number tuples due to some real-life restriction; for example some social networks impose limits on friend numbers.
In \cite{sarma2013consistent}, the thinning problem of efficiently sampling regions from a geographical dataset for visualization on a map is studied.

Motivated by the $\sqcap$-tractability theory and the observation of the size-restricted outputs of some preprocessing functions, we introduce a novel notion of $\sqcap'$-tractability and discuss corresponding complexity classes in this paper.
The main difference from the $\sqcap$-tractability theory is that we place a restriction on preprocessing functions, which limits the functions to produce a relatively small database as output, at most polylog-size of the input database.
At the same time, we bound the redundancy information when re-factorizing data and queries for preprocessing.
These changes aim to make our theory more closely linked to practice.
We define the complexity class $\sqcap'T_Q^0$, consisting of all $\sqcap'$-tractable queries, and show that the NC reduction, which is defined in the $\sqcap$-tractability theory and does not allow re-factorization, is also proven to be transitive and compatible.
We then analysis the complete problem and the size of $\sqcap'T_Q^0$.
In particular, we discover a problem, which is in $\sqcap{}T_Q^0$ but not in $\sqcap'T_Q^0$, and thus conclude that $\sqcap'T_Q^0$ is strictly smaller than $\sqcap{}T_Q^0$.
We continue to use $\sqcap'T_P$ and $\sqcap'T_Q$ to denote problems and classes of Boolean queries that can be made $\sqcap'$-tractable, respectively.
After introducing the new notion of factorization, we develop the corresponding concepts of NC reduction for $\sqcap'T_P$ and $\sqcap'T_Q$, and prove that they are transitive and compatible in $\sqcap'T_P$ and $\sqcap'T_Q$.
Furthermore, we present the complete problems of $\sqcap'T_P$ and $\sqcap'T_Q$ and acquire their sizes.

Compared to the $\sqcap$-tractability theory, the main contribution of the paper is as follows.

  (1) We consider a restriction on preprocessing functions, such that the outputs of such functions is at most polylog-size of its input, guided by some afore-mentioned database processing technologies.
      These technologies assure a relatively small database as output and make it convenient for the analysis after preprocessing.
        With this restriction, we successfully get a smaller complexity class, that is, $\sqcap'T_Q^0\subsetneq\sqcap{}T_Q^0$ and thus, $\sqcap'T_Q^0\subsetneq\text{P}$.

  (2) We bound the redundant information in the process of factorization to be at a constant level.
        It guarantees the efficiency of the subsequent preprocessing step, which is important in the background of big data.
        In addition, it also provides help to make sure that the query part is at most polylog-size of the original instance.
        We remark that some factorizations in \cite{fan2013making}, for example $\Upsilon_L$ in the proof of Theorem 5, breach this condition.

  (3) We prove the properties of transitivity and compatibility of our first NC reduction ($\leq_{fcr}^{NC}$) in a different way from their corresponding proofs in the $\sqcap$-tractability theory, because the original proofs are not applicable to our settings.

\begin{table}
\caption{List of symbols\label{list_symbol}}
{
\begin{tabular}{ll}
\hline
Symbol & Description \\
\hline
$\Upsilon=(\pi_1,\pi_2,\rho,c)$ & factorization with constant redundancy \\
$\mathcal{Q}$ & class of Boolean queries \\
$L$ & decision problem \\
$L_\mathcal{Q}$ & decision problem for $\mathcal{Q}$ \\
$S$ & language of pairs \\
$S_\mathcal{Q}$ & language of pairs for $\mathcal{Q}$ \\
$S_{(L,\Upsilon)}$ & language of pairs for $(L,\Upsilon)$ \\
$\Pi$ & data processing function \\
$\sqcap'$-tractability & tractability after data processing with polylog-size output \\
$\sqcap'T_Q^0$ & set of class of Boolean queries that are $\sqcap'$-tractable \\
$\sqcap{}'T_P$ & set of decision problems that can be made $\sqcap'$-tractable \\
$\sqcap{}'T_Q$ & set of class of Boolean queries that can be made $\sqcap'$-tractable \\
$\leq{}^{NC}_{fcr}$ & NC-fcr Reducibility \\
$\leq{}^{NC}_F$ & NC-F Reducibility \\
\hline

\end{tabular}}
\end{table}

The remainder of the paper is structured as follows.
After reviewing some concepts of the $\sqcap$-tractability theory in Section~\ref{section_background}, we introduce $\sqcap'$-tractable queries and three complexity classes $\sqcap{}'T_Q^0$, $\sqcap{}'T_P$, and $\sqcap{}'T_Q$ in Section~\ref{section_tractablequeries}.
We propose our first NC reduction and show that it is transitive and compatible in Section~\ref{section_reducibility}.
Section~\ref{section_sizecomplexityclass} is devoted to the complete problems and the sizes of $\sqcap{}'T_P$ and $\sqcap{}'T_Q$, and in Section~\ref{section_nofactorization}, we discuss the properties of $\sqcap'T_Q^0$.
The paper is concluded in Section~\ref{section_conclusion}.
For the convenience of the reader, we list the symbols and their descriptions in Table~\ref{list_symbol}.

\section{Review of Fan et al.'s Tractable Queries on Big Data}
\label{section_background}
For the convenience of the reader, we briefly review Fan et al.'s $\sqcap$-tractability theory of tractable queries on big data with preprocessing in this section. We refer the reader to \cite{fan2013making} for justification and examples.

Let us begin with the complexity classes P and NC. Recall that the complexity class P consists of all decision problems that can be solved by a deterministic Turing machine in polynomial time (PTIME), i.e., in $n^{O(1)}$ time, where $n$ is the size of the input; the parallel complexity class NC (the abbreviation of Nick's Class) consists of all decision problems that can be solved by taking $\log^{O(1)}n$ time on a PRAM with $n^{O(1)}$ processors \cite{greenlaw1995limits}.

Following the convention of complexity theory, we use a finite alphabet $\Sigma$ of symbols to encode both data and queries. In this way, a database as well as a query can be encoded into a string in $\Sigma^*$ with necessary delimiters. Recall that a language $S$ of pairs is a subset of $\Sigma^*\times\Sigma^*$, and one can use $S$ to encode a class $\mathcal{Q}$ of Boolean queries such that for each $\langle{}D,Q\rangle\in{}S$, $Q$ is a query in $\mathcal{Q}$, $D$ is a database on which $Q$ is defined, and $Q(D)$ returns true. We say that a language $S$ of pairs is in NC if it is in NC to decide whether a pair $\langle{}D,Q\rangle$ belongs to $S$ \cite{fan2013making}. Further, to denote tractable queries with preprocessing, Fan et al. introduced the following notion.

\begin{definition}[\cite{fan2013making}]\label{definition_tractable_old}
A language $S$ of pairs is \textit{$\sqcap$-tractable} if there exist a PTIME preprocessing function $\Pi:\Sigma^*\rightarrow\Sigma^*$ and a language $S'$ of pairs, which is in NC, such that for all $D,Q\in\Sigma^*$,
$$\langle{}D,Q\rangle\in{}S \mbox{ iff }\langle{}\Pi(D),Q\rangle\in{}S'.$$
A class $\mathcal{Q}$ of Boolean queries is called \textit{$\sqcap$-tractable} if the language $S_\mathcal{Q}$ of pairs for $\mathcal{Q}$ is $\sqcap$-tractable.
We use $\sqcap{}T_Q^0$ to denote the set of all $\sqcap$-tractable classes of Boolean queries.
\end{definition}

Before extending the notion of $\sqcap$-tractability to decision problems, we need a concept of factorization.
\begin{definition}[\cite{fan2013making}]
\label{definition_factorization}
We say that a language $L$ \textit{can be factored} if there exist three NC computable functions $\pi_1(\cdot)$, $\pi_2(\cdot)$, and $\rho(\cdot,\cdot)$ such that $\rho(\pi_1(x),\pi_2(x))=x$ for all $x\in{}L$. We refer to $\Upsilon=(\pi_1,\pi_2,\rho)$ and $S_{(L,\Upsilon)}=\{\langle\pi_1(x),\pi_2(x)\rangle|x\in{}L\}$ as a \textit{factorization} of $L$ and the language of pairs for $(L,\Upsilon)$, respectively.
\end{definition}

The following definition tells us when decision problems can be made $\sqcap$-tractable.

\begin{definition}[\cite{fan2013making}]\label{definition_makeproblemtractable_old}
A decision problem $L$ \textit{can be made $\sqcap$-tractable} if there exists a factorization $\Upsilon=(\pi_1,\pi_2,\rho)$ of $L$ such that the language $S_{(L,\Upsilon)}$ of pairs is $\sqcap$-tractable. We use $\sqcap{}T_P$ to denote the set of all decision problems that can be made $\sqcap$-tractable.
\end{definition}

When a class of Boolean queries is not $\sqcap$-tractable, it is possible to make it $\sqcap$-tractable by changing its data and query parts \cite{fan2013making}. To this end, for any given class $\mathcal{Q}$ of Boolean queries, let us construct a decision problem by connecting its data and query parts with a delimiter $\#$:
$$L_\mathcal{Q}=\{D\#Q|\langle{}D,Q\rangle\in{}S_{\mathcal{Q}}\},$$
where $S_{\mathcal{Q}}$ is the language of pairs for $\mathcal{Q}$ and $\#$ is a symbol that has not been used before.

\begin{definition}[\cite{fan2013making}]\label{definition_makequerytractable_old}
A class $\mathcal{Q}$ of Boolean queries \textit{can be made $\sqcap$-tractable} if the decision problem $L_\mathcal{Q}$ for $\mathcal{Q}$ can be made $\sqcap$-tractable. We use $\sqcap{}T_Q$ to denote the set of all classes of Boolean queries that can be made $\sqcap$-tractable.
\end{definition}

Intuitively, $\mathcal{Q}$ is in $\sqcap{}T_Q$ if it can be re-factorized by some factorization $\Upsilon$ such that $S_{(L_{\mathcal{Q}},\Upsilon)}$ is $\sqcap$-tractable.

Furthermore, Fan et al. introduced the following form of reductions to specify transformations from one problem to another problem in $\sqcap{}T_P$.

\begin{definition}[\cite{fan2013making}]\
\begin{enumerate}
\item[$(1)$]
A decision problem $L_1$ is said to be \textit{NC-factor reducible} to another decision problem $L_2$,
denoted by $L_{1}\leq_{fa}^{NC}L_2$, if there exist factorizations $\Upsilon_1=(\pi_1^1,\pi_2^1,\rho_1)$ of $L_1$ and $\Upsilon_2=(\pi_1^2,\pi_2^2,\rho_2)$ of $L_2$, and two NC functions $\alpha(\cdot)$ and $\beta(\cdot)$, such that for all $D,Q\in\Sigma^*$,
$\langle{}D,Q\rangle\in{}S_{(L_1,\Upsilon{}_1)}$ iff $\langle\alpha(D),\beta(Q)\rangle\in{}S_{(L_2,\Upsilon{}_2)}$.
\item[$(2)$]
Given two classes $\mathcal{Q}_1$ and $\mathcal{Q}_2$ of Boolean queries, we say that $\mathcal{Q}_1$ is \textit{NC-factor reducible} to $\mathcal{Q}_2$, denoted by $\mathcal{Q}_{1}\leq_{fa}^{NC}\mathcal{Q}_2$, if $L_{\mathcal{Q}_1}\leq_{fa}^{NC}L_{\mathcal{Q}_2}$,
where $L_{\mathcal{Q}_1}$ and $L_{\mathcal{Q}_2}$ are the decision problems for $\mathcal{Q}_1$ and $\mathcal{Q}_2$, respectively.
\end{enumerate}
\end{definition}

It has been proven in \cite{fan2013making} that the reducibility relations $\leq_{fa}^{NC}$ for both decision problems and classes of Boolean queries are transitive and compatible in $\sqcap{}T_P$ and $\sqcap{}T_Q$, respectively. Notably, by using $\leq_{fa}^{NC}$ to relate a problem to the entire complexity class $\sqcap{}T_P$, Fan et al. showed that $\sqcap{}T_P$ has complete problems, which results in that all problems in P and all classes of Boolean queries in PTIME can be made $\sqcap$-tractable \cite{fan2013making}.

For $\sqcap{}T_Q^0$, the set of $\sqcap$-tractable classes of Boolean queries, Fan et al. introduced another NC reduction $\leq_{F}^{NC}$.
\begin{definition}[\cite{fan2013making}]
\label{definition_ncfreduction}
A language $S_1$ of pairs is said to be \textit{F-reducible} to another language $S_2$ of pairs,
denoted by $S_1\leq_{F}^{NC}S_2$, if there exist NC functions $\alpha(\cdot)$ and $\beta(\cdot)$, such that for all $D,Q\in\Sigma^*$,
$\langle{}D,Q\rangle\in{}S_1$ iff $\langle\alpha(D),\beta(Q)\rangle\in{}S_2$.

A class $\mathcal{Q}_1$ of Boolean queries is said to be \textit{F-reducible} to another class $\mathcal{Q}_2$ of Boolean queries,
denoted by $\mathcal{Q}_1\leq_{F}^{NC}\mathcal{Q}_2$, if $S_1\leq_{F}^{NC}S_2$, where $S_1$ and $S_2$ are the language of pairs for $\mathcal{Q}_1$ and $\mathcal{Q}_2$, respectively.
\end{definition}

This reduction is transitive and compatible in $\sqcap{}T_Q^0$ \cite{fan2013making}.
However, the complete problem for $\sqcap{}T_Q^0$ under $\leq_{F}^{NC}$ is hard to find.
The existence of such a complete problem is close related to the open problem whether $\text{P}=\text{NC}$.
In addition, by considering a given class of Boolean queries, it has been proven that $\sqcap{}T_Q^0$ is properly contained in P, unless $\text{P}=\text{NC}$ \cite{fan2013making}.

\section{$\sqcap'$-Tractable Queries}
\label{section_tractablequeries}
In this section, we first introduce the notion of $\sqcap'$-tractability and explain its difference from $\sqcap$-tractability in Section \ref{section_tractable}, and then show how we can make problems and classes of Boolean queries $\sqcap'$-tractable via factorization with constant redundancy in Sections \ref{section_makeproblem} and \ref{section_makequery}, respectively.

Let us begin with the following convention. In our framework, we only study the languages of pairs with the following property (we will call it \textit{short-query property} in the rest of the paper): for all $\langle{}D,Q\rangle\in{}S$,
\begin{equation}
\label{equation_shortquery}
  |Q|\in{}O(\text{polylog}(|D|)).
\end{equation}

This restriction is rather reasonable, since most of the queries that we could use in practice and study in academe are much shorter than the database. For example, in the online search engine Google, the length of search query is strictly limited.
One search query can only have at most 150 query terms, which is at most as long as 128 characters \cite{google2014}.
Clearly, such a query is extremely short when compared to the large set of websites.
For academic world, on can refer to BDS problem in Example~\ref{example_bds}.
In the problem, the data part is a whole undirected graph while the query is only two vertexes.
Clearly, the query is also extremely short in this example.
Furthermore, we remark that the language $S_{\text{CVP}}$ of pairs used in $\sqcap$-tractability to prove the separation of $\sqcap{}T_Q^0$ and P does not meet the condition (Proof of Theorem 9 in \cite{fan2013making}).

\subsection{$\sqcap'$-Tractability}
\label{section_tractable}

In this section, we will explain the concept of $\sqcap'$-tractability and illustrate it by two examples.
\begin{definition}
A language $S$ of pairs, satisfying $|Q|\in{}O(\text{polylog}(|D|))$, is \textit{$\sqcap'$-tractable}
if there exist a preprocessing function $\Pi(\cdot)$ and a language $S'$ of pairs which is in NC, such that for all $D,Q\in\Sigma^*$,
\begin{enumerate}
  \item[$(1)$] $\langle{}D,Q\rangle\in{}S$ iff $\langle{}\Pi(D),Q\rangle\in{}S'$,
  \item[$(2)$] $\Pi$ is a PTIME function and for any $x$, $|\Pi(x)|\in{}O(\text{polylog}(|x|))$.
\end{enumerate}
We say that a class $\mathcal{Q}$ of Boolean queries is \textit{$\sqcap'$-tractable} if $S_\mathcal{Q}$, the language of pairs for $\mathcal{Q}$, is $\sqcap'$-tractable.
We define complexity class $\sqcap'T_Q^0$ to denote all $\sqcap'$-tractable classes of Boolean queries.
\end{definition}

On the definition, we have the following two remarks.
\begin{remark}
Compared to Definition~\ref{definition_tractable_old}, we place an interesting condition on the preprocessing function $\Pi$ that the size of its output is required not to be long, at most poly-logarithmic of its input.
As mentioned in Introduction, we add this condition to simulate the technologies used to deal big data, where they utilize a preprocessing function to generate a small database that returns the same result as the original big database.
Therefore, it is easy to see that $\sqcap'T_Q^0$ is contained in $\sqcap{}T_Q^0$.
We will disclose that the containment is proper, namely $\sqcap'T_Q^0\subsetneq\sqcap{}T_Q^0$, in Section~\ref{section_nofactorization}.
\end{remark}
\begin{remark}
Let us focus on the condition that $\langle{}D,Q\rangle\in{}S$ iff $\langle{}\Pi(D),Q\rangle\in{}S'$.
After we get a relative small database, we gain the pair $\langle{}\Pi(D),Q\rangle$.
If queries are too long, say as long as $\Theta(|D|)$, then the size of $\Pi(D)$ is negligible when compared to $Q$, and our length restriction of function $\Pi$ will be ineffective.
This is why we need the short-query property stated in Equation (\ref{equation_shortquery}).
\end{remark}

We now illustrate our new concept of $\sqcap'$-tractability by two examples.

\begin{example}
\label{example_wordstatistics}
Consider the following class $\mathcal{Q}$ of Boolean queries: A query $Q=(p,k)\in\mathcal{Q}$ on a full-length novel $\mathcal{N}$ is to find whether this preposition $p$ appears at least $k$ times in this novel $\mathcal{N}$.

We know that the usage frequency of preposition is quite different in different writers' novels.
So it is an effective way to determine the author of an anonymous novel by adding up preposition quantities.
This is the Boolean form of this query.
We are going to show that this query $\mathcal{Q}$ is $\sqcap'$-tractable.

First we get the language $S_\mathcal{Q}$ of pairs for $\mathcal{Q}$ as follow:
\begin{equation*}
  S_\mathcal{Q}=\{\langle\mathcal{N},(p,k)\rangle|\mathcal{N}\text{ is a novel and preposition }p\text{ appears at least }k\text{ times in }\mathcal{N}\}.
\end{equation*}
We know that preposition words are finite in English.
Supposing that there are totally $m$ preposition words, we construct a function $\Pi(\cdot)$ as follows: $\Pi(\mathcal{N})$ counts each preposition word in the novel $\mathcal{N}$ and returns these $m$ numbers.
We construct a language $S'$ of pairs in the following way:
\begin{equation*}
  S'=\{\langle\mathcal{L},(p,k)\rangle|\mathcal{L}\text{ is a list of }m\text{ numbers and its number for preposition }p\text{ is }\geq{}k\}.
\end{equation*}

It is easy to see that $\Pi(\cdot)$ is a PTIME function, $S'$ is in NC, and $\langle\mathcal{N},(p,k)\rangle\in{}S_\mathcal{Q}$ iff
$\langle\Pi(\mathcal{N}),(p,k)\rangle\in{}S'$. Next we will show that $\Pi(\cdot)$ satisfies the property $|\Pi(x)|\in{}O(\text{polylog}(|x|))$.
Suppose that $\mathcal{N}$ has $n$ words totally.
Then $|\mathcal{N}|\in\Theta(n)$, and each preposition word appears at most $n$ times.
We can store all these $m$ quantities by $m\text{log}(n)$ bits, i.e., the size of $\mathcal{L}$ is at most $m\text{log}(n)$.
Since $m$ is a constant number, we have $|\mathcal{L}|\in{}O(\text{polylog}(|\mathcal{N}|))$.
Consequently, $S_\mathcal{Q}$ is $\sqcap'$-tractable and $\mathcal{Q}\in\sqcap'T_Q^0$.
\end{example}

The next example is the query version of Circuit Value Problem (CVP) \cite{greenlaw1995limits}.
The example plays an important role in the subsequent analysis about our new complexity classes.

\begin{example}
\label{example_cvp}
Consider the query version of CVP:
A query $Q=\epsilon\in\mathcal{Q}_\text{CVP}$ on an instance $q$ of CVP is to find whether the output of $q$ is true.

An instance of CVP is an encoding $\langle\alpha\rangle$ of a Boolean circuit $\alpha$, inputs $x_1,x_2,\ldots,x_n$, and a designated output $y$.
Informally, a Boolean circuit $\alpha$ is a directed acyclic graph, in which a node can be
\begin{enumerate}
\item an input node $x_i$ for $i\in[1,n]$, with in-degree $0$,
\item an output node $y$ with out-degree 0,
\item a gate denoting a Boolean operator, e.g., $\neg,\wedge,\vee$, which applies to its input (child) and feeds the result as an input for its parent node.
\end{enumerate}
It is a Boolean function that takes inputs $x_1,x_2,\ldots,x_n$ and returns the truth value of $y$. Its encoding $\langle\alpha\rangle$ is a sequence of tuples, one for each node in the directed acyclic graph.

It has been known that CVP is a P-complete problem \cite{greenlaw1995limits}, so there exists a PTIME function to solve CVP.
We set $\Pi$ to be this function, i.e., $\Pi(q)$ returns $1$ if the output of $q$ is true and returns $0$, otherwise.
We set the language $S'$ of pairs to be $S'=\{\langle{}1,\epsilon\rangle\}$.
It is obvious that $S'$ is in NC, $\Pi$ is a PTIME function with polylog-size output, and
$\langle{}q,\epsilon\rangle\in{}S_\text{CVP}$ iff $\langle{}\Pi(q),\epsilon\rangle\in{}S'$.
So $\mathcal{Q}_\text{CVP}$ is also $\sqcap'$-tractable and $\mathcal{Q}_\text{CVP}\in\sqcap'T_Q^0$.

\end{example}

\subsection{Making problems $\sqcap'$-tractable}
\label{section_makeproblem}

Similar to the $\sqcap$-tractability, we are going to extend the notion of $\sqcap'$-tractability to decision problems in this section.

It is also necessary to factorize a decision problem into a data part and a query part. Let us begin with the following notion. Although we still use the term of factorization, it is imposed on an interesting restriction.

\begin{definition}
We say that a language $L$ has a \textit{factorization with constant redundancy} (or simply cr-factorization) if there exist three NC functions $\pi_1(\cdot)$, $\pi_2(\cdot)$, $\rho(\cdot,\cdot)$, and one constant number $c$ such that for all $x\in{}L$,
\begin{enumerate}
\item[$(1)$] $\rho(\pi_1(x),\pi_2(x))=x$,
\item[$(2)$] $|\pi_1(x)|+|\pi_2(x)|\leq|x|+c$,
\item[$(3)$] $|\pi_2(x)|\in{}O(\text{polylog}(|\pi_1(x)|))$.
\end{enumerate}
We refer to $\Upsilon=(\pi_1,\pi_2,\rho,c)$ as a factorization of $L$ with constant redundancy and refer to $S_{(L,\Upsilon)}=\{\langle\pi_1(x),\pi_2(x)\rangle|x\in{}L\}$ as the language of pairs for $(L,\Upsilon)$.
\end{definition}

Intuitively, such a factorization with constant redundancy can separate each instance $x$ of $L$ into data part $\pi_1(x)$ and query part $\pi_2(x)$.
The function $\rho(\cdot,\cdot)$ is an inverse function that can restore the original instance $x$ to make sure that no information is lost in the producer of factorization.
Compared to Definition \ref{definition_factorization}, the functions $\pi_1$, $\pi_2$, and $\rho$ have the same meanings, but the constant number $c$ is a new component, which is a restriction on redundancy,
and the size restriction of $\pi_1(x)$ and $\pi_2(x)$ is also new added.
Since we only study the queries with short-query property stated in Equation~(\ref{equation_shortquery}) and we separate the instance $x$ into
$\pi_1(x)$ and $\pi_2(x)$ to treat them as data part and query part, it is nature to consider restriction (3) here.

The reasons why we consider the constant number $c$ are as follows.
First, comparing the information volume of $\pi_1(x)$, $\pi_2(x)$, and $x$, we can see that all information in $x$ has been arranged into $\pi_1(x)$ and $\pi_2(x)$, otherwise $\rho$ cannot work.
Additionally, we also need to require that the sizes of $\pi_1(x)$ and $\pi_2(x)$ should not be too large, because in the context of big data, even linear-increase is expensive. We have to limit the redundancy information at a constant level, in case that the factorization affects the efficiency of preprocessing.
Second, only given restriction (3), we cannot limit the size of $\pi_2(x)$ when compared to the size of $x$.
With the help of $c$,
we can get that $|\pi_2(x)|$ is also in $O(\text{polylog}(|x|))$, as shown in the following proposition.
This proposition will facilitate the subsequent investigations of our complexity classes.

\begin{proposition}
\label{proposition_pi2size}
  Given a language $L$ and a factorization of $L$ with constant redundancy $\Upsilon=(\pi_1,\pi_2,\rho,c)$,
  it holds that $|\pi_2(x)|\in{}O(\text{polylog}(|x|))$ for any $x\in{}L$.
\end{proposition}
\begin{proof}
  Since $|\pi_1(x)|+|\pi_2(x)|\leq|x|+c$, we have that $|\pi_1(x)|\leq|x|+c$.
  Additionally, we have that $|\pi_2(x)|\in{}O(\text{polylog}(|\pi_1(x)|))$, so $|\pi_2(x)|\in{}O(\text{polylog}(|x|+c))=O(\text{polylog}(|x|))$, as desired.
\end{proof}

Similar to Definition \ref{definition_makeproblemtractable_old}, we have the following.

\begin{definition}
We say that a decision problem $L$ \textit{can be made $\sqcap'$-tractable} if there exists a factorization of $L$ with constant redundancy $\Upsilon=(\pi_1,\pi_2,\rho,c)$ such that the language $S_{(L,\Upsilon)}$ of pairs is $\sqcap'$-tractable.
We write $\sqcap{}'T_P$ for the set of all decision problems that can be made $\sqcap'$-tractable.
\end{definition}

Let us illustrate the above concepts with an example.

\begin{example}
\label{example_bds}
Consider the decision problem of Breadth-Depth Search (BDS) \cite{greenlaw1995limits}:
\begin{itemize}
\item Input: An undirected graph $G=(V,E)$ with a number on every node, and a pair $(u,v)$ of nodes in $V$.
\item Question: Is $u$ visited before $v$ in the breadth-depth search of $G$ induced by the vertex numbering?
\end{itemize}
A breadth-depth search starts at a node $s$ with the minimum number, and visits all its children, pushing them onto a stack in the reverse order induced by the vertex numbering as the search proceeds. After all children of $s$ are visited, the search continues with the node on the top of the stack, which plays the role of $s$.

Now, we are going to show that BDS can be made $\sqcap'$-tractable. To this end, we define a cr-factorization $\Upsilon=(\pi_1,\pi_2,\rho,c)$ of BDS: For all $x=(G,(u,v))\in\text{BDS}$, set $\pi_1(x)=(G,(u,v))$ and $\pi_2(x)=\epsilon$;
for any graph $G$ and any pair $(u,v)$ of nodes, let $\rho((G,(u,v)),\epsilon)=(G,(u,v))$, and then set $c=0$.

Then we can obtain that
\begin{equation*}
\begin{aligned}
S_{(\text{BDS},\Upsilon)}=&\{\langle{}(G,(u,v)),\epsilon\rangle|G\text{ is a graph and }u\text{ is visited before }v\text{ in the breadth-depth search of }G\}.
\end{aligned}
\end{equation*}

Just like CVP in Example~\ref{example_cvp}, BDS is also a P-complete problem \cite{greenlaw1993breadth},
so we can also get a PTIME function to solve BDS.
We construct the preprocessing function $\Pi$ to be this function and the new language $S'$ of pairs to be $\{\langle{}1,\epsilon\rangle\}$, in the same way as Example~\ref{example_cvp},
and also conclude that:
$\langle{}(G,(u,v)),\epsilon\rangle\in{}S_{(\text{BDS},\Upsilon)}$ iff $\langle{}\Pi(G,(u,v)),\epsilon\rangle\in{}S'$.
So $S_{\text{BDS}}$ is $\sqcap'$-tractable, and BDS can be made $\sqcap'$-tractable, i.e., $\text{BDS}\in\sqcap'T_P$.
\end{example}

\subsection{Making classes of Boolean queries $\sqcap'$-tractable}
\label{section_makequery}
Analogue to the $\sqcap$-tractability, for a class of Boolean queries that is not $\sqcap'$-tractable itself, it is possible to make the class $\sqcap'$-tractable
by changing its data and query parts.
We discuss how can we make queries $\sqcap'$-tractable in this section.

Given a class $\mathcal{Q}$ of Boolean queries, we construct a decision problem
$$L_\mathcal{Q}=\{D\#Q|\langle{}D,Q\rangle\in{}S_{\mathcal{Q}}\},$$
where $S_{\mathcal{Q}}$ is the language of pairs for $\mathcal{Q}$. The problem $L_\mathcal{Q}$ can be referred to as the decision problem for $\mathcal{Q}$ and can be stated as:
\begin{itemize}
\item Input: A string in $\Sigma^*$ with the form $D\#Q$.
\item Question: Does $Q(D)$ return true?
\end{itemize}

Corresponding to Definition~\ref{definition_makequerytractable_old}, we have the following definition.
\begin{definition}
We say that a class $\mathcal{Q}$ of Boolean queries \textit{can be made $\sqcap'$-tractable} if the decision problem $L_\mathcal{Q}$ for $\mathcal{Q}$ can be made $\sqcap'$-tractable. We use $\sqcap{}'T_Q$ to denote the set of all classes of Boolean queries that can be made $\sqcap'$-tractable.
\end{definition}

\begin{figure}
\centerline{\includegraphics[width=7cm, trim=6cm 1.5cm 8cm 2cm, clip=true]{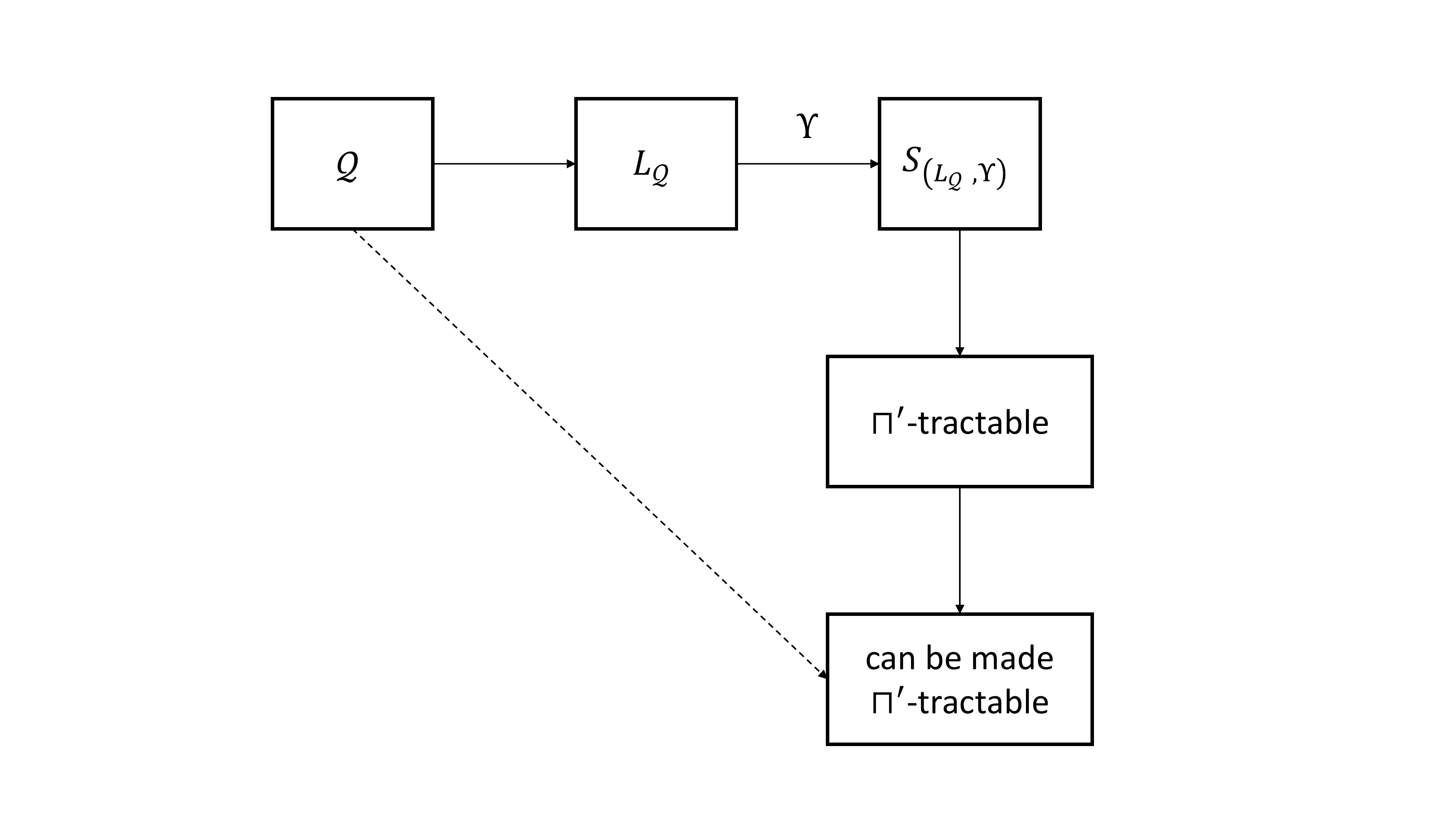}}
\caption{Making $\mathcal{Q}$ $\sqcap'$-tractable.}
\label{figure}
\end{figure}

We know that $L_\mathcal{Q}$ can be made $\sqcap'$-tractable means that there exists a cr-factorization $\Upsilon$ such that $S_{(L_\mathcal{Q},\Upsilon)}$ is $\sqcap'$-tractable.
However, it is possible that $S_{(L_\mathcal{Q},\Upsilon)}$ is not the language of pairs for $\mathcal{Q}$, that is, $\mathcal{Q}$ is in $\sqcap{}'T_Q$ if it can be re-factorized by some cr-factorization $\Upsilon$ such that $S_{(L_\mathcal{Q},\Upsilon)}$ is $\sqcap'$-tractable.
Figure~\ref{figure} depicts the process.

We end this section with an example arising from Example~\ref{example_bds}.

\begin{example} \label{example_bdsquery}
Consider the query version of BDS in Example~\ref{example_bds}:
A query $Q=(u,v)\in\mathcal{Q}_\text{BDS}$ on a graph $G$ is to find whether $u$ is visited before $v$ in the breadth-depth search of $G$ induced by the vertex numbering.

We are going to show that $\mathcal{Q}_\text{BDS}$ can be made $\sqcap'$-tractable. When we get the decision problem
\begin{equation*}
\begin{aligned}
L_{\mathcal{Q}_\text{BDS}}=&\{G\#(u,v)|G\text{ is a graph and }u\text{ is visited before }v\text{ in the breadth-depth search of }G\}
\end{aligned}
\end{equation*}
for $\mathcal{Q}_\text{BDS}$,  we will discover that it is quite similar to our decision problem in Example~\ref{example_bds}. So we can make a little change on the cr-factorization defined in Example~\ref{example_bds} such that it becomes efficient to demonstrate that $\mathcal{Q}_\text{BDS}$ can be made $\sqcap'$-tractable.
We define a cr-factorization $\Upsilon'=(\sigma_1,\sigma_2,\rho',c')$ of $L_{\mathcal{Q}_\text{BDS}}$:
For all $y=G\#(u,v)\in{}L_{\mathcal{Q}_\text{BDS}}$, set $\sigma_1(y)=(G,(u,v))$, $\sigma_2(y)=\epsilon$;
for any graph $G$ and any pair $(u,v)$ of nodes, let $\rho'((G,(u,v)),\epsilon)=G\#(u,v)$; and take $c'=-1$.
Then we find that $S_{(L_{\mathcal{Q}_\text{BDS}},\Upsilon')}$ equals $S_{(\text{BDS},\Upsilon)}$ in Example~\ref{example_bds}.
Since $S_{(\text{BDS},\Upsilon)}$ in Example~\ref{example_bds} is $\sqcap'$-tractable,
$S_{(L_{\mathcal{Q}_\text{BDS}},\Upsilon')}$ is also $\sqcap'$-tractable and $\mathcal{Q}_\text{BDS}$ can be made $\sqcap'$-tractable, i.e., $\mathcal{Q}_\text{BDS}\in\sqcap{}'T_Q$.

\end{example}

\section{NC-fcr Reducibility}
\label{section_reducibility}
Having defined the complexity classes $\sqcap{}'T_P$, $\sqcap{}'T_Q$, and $\sqcap{}'T_Q^0$, we have several nature questions about them.
How large are these complexity classes? Is there an efficient solution to decide whether a given decision problem is in $\sqcap{}'T_P$ and whether a given class of Boolean queries is in $\sqcap{}'T_Q$ or $\sqcap{}'T_Q^0$?
In this section, we define a reduction and then show that it is compatible in our $\sqcap'T_P$ and $\sqcap'T_Q$.
We differ exploring the properties of $\sqcap'T_P$ and $\sqcap'T_Q$ to Section~\ref{section_sizecomplexityclass}
and the properties of $\sqcap{}'T_Q^0$ to Section~\ref{section_nofactorization}.
In Section~\ref{section_sizecomplexityclass}, we will introduce a complete problem such that we can decide whether a given decision problem is in $\sqcap{}'T_P$ by trying to reduce it to the complete problem, and similarly for classes of Boolean queries.

Due to we have adopted the new concept of factorization, we have to re-formalize a corresponding definition of reducibility.
\begin{definition}
\label{definition_ncfcrreduction}\ 
\begin{enumerate}
\item[$(1)$]
A decision problem $L_1$ is said to be \textit{NC-fcr reducible} to another decision problem $L_2$,
denoted by $L_{1}\leq_{fcr}^{NC}L_2$, if there exist factorizations $\Upsilon_1=(\pi_1^1,\pi_2^1,\rho_1,c_1)$ of $L_1$ and $\Upsilon_2=(\pi_1^2,\pi_2^2,\rho_2,c_2)$ of $ L_2 $ with constant redundancy, and two NC functions $\alpha(\cdot)$ and $\beta(\cdot)$, such that for all $D,Q\in\Sigma^*$,
$\langle{}D,Q\rangle\in{}S_{(L_1,\Upsilon{}_1)}$ iff $\langle\alpha(D),\beta(Q)\rangle\in{}S_{(L_2,\Upsilon{}_2)}$.
\item[$(2)$]
For two classes $\mathcal{Q}_1$ and $\mathcal{Q}_2$ of Boolean queries, we say that $\mathcal{Q}_1$ is \textit{NC-fcr reducible} to $\mathcal{Q}_2$, denoted by $\mathcal{Q}_1\leq_{fcr}^{NC}\mathcal{Q}_2$, if $L_{\mathcal{Q}_1}\leq_{fcr}^{NC}L_{\mathcal{Q}_2}$,
where $L_{\mathcal{Q}_1}$ and $L_{\mathcal{Q}_2}$ are decision problems for $\mathcal{Q}_1$ and $\mathcal{Q}_2$, respectively.
\end{enumerate}
\end{definition}

As mentioned in Section~\ref{section_background}, Fan et al. \cite{fan2013making} showed the transitive and compatible properties of the reducibility relation $\leq_{fa}^{NC}$ for both decision problems and classes of Boolean queries.
For our reducibility $\leq_{fcr}^{NC}$ in the $\sqcap'$-tractability, it turns out that the properties still hold. Unfortunately, the proof methods in \cite{fan2013making} are not applicable to our framework. Let us state the properties and update the proofs.

\begin{theorem}
\label{theorem_transitive}
The NC-fcr reduction $\leq_{fcr}^{NC}$ for decision problems is transitive, i.e., if $L_{1}\leq_{fcr}^{NC}L_2$ and $L_{2}\leq_{fcr}^{NC}L_3$, then $L_{1}\leq_{fcr}^{NC}L_3$.
\end{theorem}
\begin{proof}
Given that $L_{1}\leq_{fcr}^{NC}L_2$ and $L_{2}\leq_{fcr}^{NC}L_3$, we have the cr-factorizations:
$\Upsilon_1=(\pi_1^1,\pi_2^1,\rho_1,c_1)$ of $L_1$,
$\Upsilon_2=(\pi_1^2,\pi_2^2,\rho_2,c_2)$ and $\Upsilon'_2=(\sigma_1^2,\sigma_2^2,\rho'_2,c'_2)$ of $L_2$,
$\Upsilon_3=(\pi_1^3,\pi_2^3,\rho_3,c_3)$ of $L_3$,
and four NC functions $\alpha_1(\cdot),\alpha_2(\cdot),\beta_1(\cdot),\beta_2(\cdot)$, such that for all $D,Q\in\Sigma^*$,
\begin{equation}
\label{equation_transitive1}
\langle{}D,Q\rangle\in{}S_{(L_1,\Upsilon_1)}
\text{ iff }\langle{}\alpha_1(D),\beta_1(Q)\rangle\in{}S_{(L_2,\Upsilon_2)},
\end{equation}
\begin{equation}
\label{equation_transitive2}
\langle{}D,Q\rangle\in{}S_{(L_2,\Upsilon'_2)}
\text{ iff }\langle{}\alpha_2(D),\beta_2(Q)\rangle\in{}S_{(L_3,\Upsilon_3)}.
\end{equation}

We now construct two cr-factorizations $\Upsilon'_1=(\sigma_1^1,\sigma_2^1,\rho'_1,c'_1)$ of $L_1$ and $\Upsilon'_3=(\sigma_1^3,\sigma_2^3,\rho'_3,c'_3)$ of $L_3$ by setting
\begin{eqnarray*}
\sigma_1^1(x)\!\!\!&=&\!\!\!\pi_1^1(x)@\pi_2^1(x), \\
 \sigma_2^1(x) \!\!\!&=&\!\!\! \epsilon, \\
  \sigma_1^3(x) \!\!\!&=&\!\!\!\pi_1^3(x)@\pi_2^3(x), \\
\sigma_2^3(x) \!\!\!&=&\!\!\!\epsilon,
\end{eqnarray*}
for all $x\in\Sigma^*$, setting
\begin{eqnarray*}
\rho'_1(x_1@x_2,\epsilon)\!\!\!&=&\!\!\!\rho_1(x_1,x_2), \\
\rho'_3(x_1@x_2,\epsilon)\!\!\!&=&\!\!\!\rho_3(x_1,x_2),
\end{eqnarray*}
for all $x_1,x_2\in\Sigma^*$, and taking $c'_1=c_1+1$ and $c'_3=c_3+1$, where the symbol $@$ is a new symbol that has never been used before. For convenience, we define two auxiliary functions
\begin{eqnarray*}
f(x@y)\!\!\!&=&\!\!\! x, \\
g(x@y)\!\!\!&=&\!\!\! y
\end{eqnarray*}
to extract the front and back parts of a string containing the symbol $@$.

For any $\langle{}D,Q\rangle\in{}S_{(L_1,\Upsilon'_1)}$, it follows by the definition of $\Upsilon'_1$ that
\begin{equation*}
Q=\epsilon\text{ and }\langle{}f(D),g(D)\rangle\in{}S_{(L_1,\Upsilon_1)}.
\end{equation*}
Hence, we see from Equation~(\ref{equation_transitive1}) that
\begin{equation*}
\langle{}\alpha_1(f(D)),\beta_1(g(D))\rangle\in{}S_{(L_2,\Upsilon_2)}.
\end{equation*}
It yields that
\begin{equation*}
\rho_2(\alpha_1(f(D)),\beta_1(g(D)))\in{}L_2,
\end{equation*}
according to the definition of $\Upsilon_2$. Setting $h(x)=\rho_2(\alpha_1(f(x)),\beta_1(g(x)))$ for simplicity, we get that $h(D)\in{}L_2$, which means that
\begin{equation*}
\langle{}\sigma_1^2(h(D)),\sigma_2^2(h(D))\rangle\in{}S_{(L_2,\Upsilon'_2)}
\end{equation*}
by the definition of $S_{(L_2,\Upsilon'_2)}$. It follows from Equation~(\ref{equation_transitive2}) that
\begin{equation*}
\langle{}\alpha_2(\sigma_1^2(h(D))),\beta_2(\sigma_2^2(h(D)))\rangle\in{}S_{(L_3,\Upsilon_3)}.
\end{equation*}
Further, by the definition of $\Upsilon'_3$, we obtain that
\begin{equation*}
\langle\alpha_2(\sigma_1^2(h(D)))@\beta_2(\sigma_2^2(h(D))),\epsilon\rangle\in{}S_{(L_3,\Upsilon'_3)}.
\end{equation*}
Defining
\begin{equation*}
\begin{aligned}
\alpha(x)&=\alpha_2(\sigma_1^2(h(x)))@\beta_2(\sigma_2^2(h(x))), \\
\beta(x)&=x,
\end{aligned}
\end{equation*}
we thus get that
$\langle\alpha(D),\beta(Q)\rangle\in{}S_{(L_3,\Upsilon'_3)}$. As a result, we have checked that $\langle\alpha(D),\beta(Q)\rangle\in{}S_{(L_3,\Upsilon'_3)}$ if $\langle{}D,Q\rangle\in{}S_{(L_1,\Upsilon'_1)}$. Conversely, there is no difficulty to derive $\langle{}D,Q\rangle\in{}S_{(L_1,\Upsilon'_1)}$ from $\langle\alpha(D),\beta(Q)\rangle\in{}S_{(L_3,\Upsilon'_3)}$ by rollback. Therefore, we have that $\langle{}D,Q\rangle\in{}S_{(L_1,\Upsilon'_1)}$ iff $\langle\alpha(D),\beta(Q)\rangle\in{}S_{(L_3,\Upsilon'_3)}$. It is easy to see that $\alpha(\cdot)$ and $\beta(\cdot)$ are NC functions. We thus obtain that $L_{1}\leq_{fcr}^{NC}L_3$, finishing the proof.
\end{proof}

After showing that $\leq_{fcr}^{NC}$ is transitive, we are now in a position to verify that $\leq_{fcr}^{NC}$ is compatible with $\sqcap{}'T_P$.
The significance of this property lies in that it shows us how to make a decision problem $\sqcap'$-tractable:
For a given decision problem $L$, we can try to find another decision problem $L'$ that can be made $\sqcap'$-tractable, and then prove that $L\leq_{fcr}^{NC}L'$.
The properties also tell us how to show that a decision problem $L$ cannot be made $\sqcap'$-tractable:
We may find another decision problem $L'$ that cannot be made $\sqcap'$-tractable, and then show that
$L'\leq_{fcr}^{NC}L$.
\begin{theorem}
\label{theorem_compatible}
The reduction $\leq_{fcr}^{NC}$ is compatible with $\sqcap{}'T_P$, that is, if $L_{1}\leq_{fcr}^{NC}L_2$ and $L_2$ can be made $\sqcap'$-tractable, then $L_1$ can also be made $\sqcap'$-tractable.
\end{theorem}
\begin{proof}
Let $L_{1}\leq_{fcr}^{NC}L_2$. Then there exist cr-factorizations
$\Upsilon_1=(\pi_1^1,\pi_2^1,\rho_1,c_1)$ of $L_1$ and $\Upsilon_2=(\pi_1^2,\pi_2^2,\rho_2,c_2)$ of $L_2$
and two NC functions $\alpha(\cdot)$ and $\beta(\cdot)$ such that for all $D,Q\in\Sigma^*$,
\begin{equation}
\label{equation_compatible1}
\langle{}D,Q\rangle\in{}S_{(L_1,\Upsilon_1)}\text{ iff }
\langle{}\alpha(D),\beta(Q)\rangle\in{}S_{(L_2,\Upsilon_2)}.
\end{equation}
Assuming that $L_{2}$ can be made $\sqcap'$-tractable, we have a cr-factorization
$\Upsilon'_2=(\sigma_1^2,\sigma_2^2,\rho'_2,c'_2)$ of $L_2$ such that $S_{(L_2,\Upsilon'_2)}$ is $\sqcap'$-tractable.
Whence, there are a PTIME function $\Pi(\cdot)$ and a language $S$ of pairs which is in NC, such that for all $D,Q\in\Sigma^*$,
\begin{equation}
\label{equation_compatible2}
\langle{}D,Q\rangle\in{}S_{(L_2,\Upsilon'_2)}
\text{ iff }\langle{}\Pi(D),Q\rangle\in{}S.
\end{equation}

We now construct a cr-factorization $\Upsilon'_1=(\sigma_1^1,\sigma_2^1,\rho'_1,c'_1)$ of $L_1$ by letting
\begin{eqnarray*}
  &&\sigma_1^1(x)=\pi_1^1(x)@\pi_2^1(x), \\
  &&\sigma_2^1(x)=\epsilon, \\
  &&\rho'_1(x_1@x_2,\epsilon)=\rho_1(x_1,x_2), \\
  &&c'_1=c_1+1,
\end{eqnarray*}
where $x, x_1,x_2\in\Sigma^*$ and the symbol $@$ is a new symbol that has never been used before. Let $S'$ be the set of pairs
$\langle{}D,Q\rangle$ satisfying that $Q=\epsilon$ and $\langle{}f(D),g(D)\rangle\in{}S$, where the functions $f$ and $g$ are two auxiliary functions defined by $f(x@y)=x$ and $g(x@y)=y$ to extract the front and back parts of a string containing the symbol $@$.

It is easy to verify that $\Upsilon'_1$ satisfies the definition of cr-factorization and $S'$ is in NC. Moreover, for any $\langle{}D,Q\rangle\in{}S_{(L_1,\Upsilon'_1)}$, we see by the definition of $\Upsilon'_1$ that
\begin{equation*}
Q=\epsilon\text{ and }\langle{}f(D),g(D)\rangle\in{}S_{(L_1,\Upsilon_1)}.
\end{equation*}
The latter gives that
\begin{equation*}
\langle{}\alpha(f(D)),\beta(g(D))\rangle\in{}S_{(L_2,\Upsilon_2)}
\end{equation*}
by Equation~(\ref{equation_compatible1}). According to the definition of $\Upsilon_2$, we have that
\begin{equation*}
\rho_2(\alpha(f(D)),\beta(g(D)))\in{}L_2.
\end{equation*}
Letting $h(x)=\rho_2(\alpha(f(x)),\beta(g(x)))$, we see that $h(D)\in L_2$, and thus
\begin{equation*}
\langle\sigma_1^2(h(D)),\sigma_2^2(h(D))\rangle\in{}S_{(L_2,\Upsilon'_2)}
\end{equation*}
by the definition of $\Upsilon'_2$. By Equation~(\ref{equation_compatible2}), we get that
\begin{equation*}
\langle\Pi(\sigma_1^2(h(D))),\sigma_2^2(h(D))\rangle\in{}S,
\end{equation*}
which implies that
\begin{equation*}
\langle\Pi(\sigma_1^2(h(D)))@\sigma_2^2(h(D)),\epsilon\rangle\in{}S'
\end{equation*}
according to the definition of $S'$. Defining $\Pi'(x)=\Pi(\sigma_1^2(h(x)))@\sigma_2^2(h(x))$, we obtain that  $\langle\Pi'(D),Q\rangle\in{}S'$.
Conversely, if $\langle\Pi'(D),Q\rangle\in{}S'$, we can check that $\langle{}D,Q\rangle\in{}S_{(L_1,\Upsilon'_1)}$ by stepping back the previous derivation.

In addition, it is easy to see that $\Pi'(\cdot)$ is a PTIME function.
It remains to show that $\Pi'$ gets a relatively short output, i.e., $|\Pi'(x)|\in{}O(\text{polylog}(|x|))$.
In fact, it is easy to get that $|h(x)|=|\rho_2(\alpha(f(x)),\beta(g(x)))|\in{}O(\text{polylog}(|x|))$, since $\rho_2,\alpha,\beta$ are all NC functions.
On the other hand, because $\sigma'_1$ is an NC function, $|\Pi(x)|\in{}O(\text{polylog}(|x|))$, and $|\sigma^2_2(x)|\in{}O(\text{polylog}(|x|))$ by Proposition~\ref{proposition_pi2size}, we know that $|\Pi'(x)|=|\Pi(\sigma_1^2(h(x)))@\sigma_2^2(h(x))|\in{}O(\text{polylog}(|x|))$.
As a result, we see that $L_1$ can be made $\sqcap'$-tractable.
This completes the proof of the theorem.
\end{proof}

As an immediate result of Theorems \ref{theorem_transitive} and \ref{theorem_compatible}, we have the following observation.

\begin{corollary}\ 
\begin{enumerate}
\item[$(1)$] The reduction $\leq_{fcr}^{NC}$ for classes of Boolean queries is transitive, i.e., if $\mathcal{Q}_{1}\leq_{fcr}^{NC}\mathcal{Q}_2$ and $\mathcal{Q}_{2}\leq_{fcr}^{NC}\mathcal{Q}_3$, then $\mathcal{Q}_{1}\leq_{fcr}^{NC}\mathcal{Q}_3$.
\item[$(2)$] The reduction $\leq_{fcr}^{NC}$ for classes of Boolean queries is compatible in $\sqcap{}'T_Q$, i.e., if $\mathcal{Q}_{1}\leq_{fcr}^{NC}\mathcal{Q}_2$ and $\mathcal{Q}_2$ can be made $\sqcap'$-tractable, then $\mathcal{Q}_1$ can be made $\sqcap'$-tractable.
\end{enumerate}
\end{corollary}

The above corollary tells us how to indirectly show that a class $\mathcal{Q}$ of Boolean queries can be made $\sqcap'$-tractable: try to find another class $\mathcal{Q}'$ of Boolean queries that can be made $\sqcap'$-tractable and then prove that $\mathcal{Q}\leq_{fcr}^{NC}\mathcal{Q}'$. It also tells us how to show that a class $\mathcal{Q}$ of Boolean queries cannot be made $\sqcap'$-tractable: find a class $\mathcal{Q}'$ of Boolean queries that cannot be made $\sqcap'$-tractable and then prove
$\mathcal{Q}'\leq_{fcr}^{NC}\mathcal{Q}$.

\section{Complete Problems for $\sqcap{}'T_P$ and $\sqcap{}'T_Q$}
\label{section_sizecomplexityclass}

In the $\sqcap$-tractability theory, it has been shown that the decision problem BDS defined in Example~\ref{example_bds} is $\sqcap{}T_P$-complete and the class $\mathcal{Q}_{BDS}$ of Boolean queries defined in Example~\ref{example_bdsquery} is $\sqcap{}T_Q$-complete \cite{fan2013making}.
In this section, we show that under our NC-fcr reducibility, BDS is also $\sqcap{}'T_P$-complete and $\mathcal{Q}_{BDS}$ is also $\sqcap{}'T_Q$-complete, although we have made some requirements on preprocessing functions and factorizations.

In order to investigate the completeness of $\sqcap{}'T_P$ and $\sqcap{}'T_Q$, we need the following concepts.
\begin{definition}\ 
\begin{enumerate}
\item[$(1)$]
A problem $L$ is \textit{$\sqcap{}'T_P$-hard} under NC-fcr reducibility if $L'\leq_{fcr}^{NC}L$ for any decision problem $L'$ in $\sqcap{}'T_P$.
A problem $L$ is \textit{$\sqcap{}'T_P$-complete} under NC-fcr reducibility
if $L$ is $\sqcap{}'T_P$-hard and it can be made $\sqcap'$-tractable itself.
\item[$(2)$]
A class $\mathcal{Q}$ of Boolean queries is \textit{$\sqcap{}'T_Q$-hard} under NC-fcr reducibility if $\mathcal{Q}'\leq_{fcr}^{NC}\mathcal{Q}$ for any class $\mathcal{Q}'$ of Boolean queries in $\sqcap{}'T_Q$.
A class $\mathcal{Q}$ of Boolean queries is \textit{$\sqcap{}'T_Q$-complete} under NC-fcr reducibility if $\mathcal{Q}$ is $\sqcap{}'T_Q$-hard and it can be made $\sqcap'$-tractable itself.
\end{enumerate}
\end{definition}

The above definition is similar to NP-complete and captures the difficulty intrinsic of $\sqcap{}'T_P$ and $\sqcap{}'T_Q$.
Intuitively, a problem $L$ is $\sqcap{}'T_P$-complete if $L$ is no easier than any other decision problem in $\sqcap{}'T_P$
and a class $\mathcal{Q}$ of Boolean queries is $\sqcap{}'T_Q$-complete if $\mathcal{Q}$ is no easier than any other class of Boolean queries in $\sqcap{}'T_Q$.

As desired, we have the following complete problems.

\begin{theorem}\ 
\label{theorem_complete}
\begin{enumerate}
\item[$(1)$] Under NC-fcr reducibility, BDS is $\sqcap{}'T_P$-complete.
\item[$(2)$] Under NC-fcr reducibility, $\mathcal{Q}_{BDS}$ is $\sqcap{}'T_Q$-complete.
\end{enumerate}
\end{theorem}
\begin{proof} Observing that (2) follows immediately from (1), we only need to prove that BDS is $\sqcap{}'T_P$-complete.
Because we have already shown that $\text{BDS}\in\sqcap{}'T_P$ in Example~\ref{example_bds}, it is sufficient to prove that BDS is $\sqcap{}'T_P$-hard.

Thanks to $\text{BDS}\in\sqcap{}'T_P$, there exists a cr-factorization $\Upsilon=(\pi_1,\pi_2,\rho,c)$ of BDS such that $S_{(\text{BDS},\Upsilon)}$ is $\sqcap'$-tractable.
For any decision problem $L\in\sqcap{}'T_P$, we know that $L\in\text{P}$ by combining the PTIME preprocessing step and the NC query evaluation step.
Considering that BDS is a P-complete problem \cite{greenlaw1993breadth}, there exists an NC function $h(\cdot)$ such that for all $x\in\Sigma^*$,
\begin{equation}
\label{equation_sizecomplex}
x\in{}L\text{ iff }h(x)\in{}\text{BDS.}
\end{equation}
We now construct two cr-factorizations $\Upsilon_L=(\sigma_1,\sigma_2,\rho_L,c_L)$ of $L$ and $\Upsilon'=(\pi'_1,\pi'_2,\rho',c')$ of BDS by defining
\begin{eqnarray*}
   &&\sigma_1(x)=x, \\
   &&\sigma_2(x)=\epsilon, \\
   &&\rho_L(x,\epsilon)=x, \\
   &&c_L=0, \\
   &&\pi'_1(x)=\pi_1(x)@\pi_2(x), \\
   &&\pi'_2(x)=\epsilon,\\
   &&\rho'(x_1@x_2,\epsilon)=\rho(x_1,x_2), \\
   &&c'=1,
\end{eqnarray*}
for all $x,x_1,x_2\in\Sigma^*$, where the symbol $@$ is a new symbol that has never been used anywhere.

For any $\langle{}D,Q\rangle\in{}S_{(L,\Upsilon_L)}$, it follows from the definition of $\Upsilon_L$ that
$Q=\epsilon$ and $D\in{}L$.
By Equation~(\ref{equation_sizecomplex}), the latter yields that $h(D)\in{}$BDS, which gives rise to
\begin{equation*}
\langle\pi'_1(h(D)),\pi'_2(h(D))\rangle\in{}S_{(\text{BDS},\Upsilon')}
\end{equation*}
by the definition of $S_{(\text{BDS},\Upsilon')}$. We thus get by the definitions of $\pi'_1$ and $\pi'_2$ that
\begin{equation*}
\langle\pi_1(h(D))@\pi_2(h(D)),\epsilon\rangle\in{}S_{(\text{BDS},\Upsilon')}.
\end{equation*}
Letting $\alpha(x)=\pi_1(h(x))@\pi_2(h(x))$ and $\beta(x)=x$, we obtain that $\langle\alpha(D),\beta(Q)\rangle\in{}S_{(\text{BDS},\Upsilon')}$.
It is easy to see that $\alpha(\cdot)$ and $\beta(\cdot)$ are NC functions, and the above  derivation can be easily checked back. Therefore, we have that $L\leq_{fcr}^{NC}\text{BDS}$, and thus BDS is $\sqcap{}'T_P$-hard, finishing the proof of the theorem.
\end{proof}

From the proof of Theorem \ref{theorem_complete}, it is easy to make the following observation.
\begin{corollary}\ 
\label{corollary_sizecomplex}
\begin{enumerate}
\item[$(1)$] $\sqcap{}'T_P=P$.
\item[$(2)$] $\sqcap{}'T_Q=\{\mathcal{Q}|\mathcal{Q}\text{ is a PTIME class of Boolean queries}\}$.
\end{enumerate}
\end{corollary}

Recall that in the $\sqcap$-tractability theory, it is also concluded that $\sqcap{}T_P=\text{P}$ and that all classes of Boolean queries that are in PTIME can be made $\sqcap$-tractable \cite{fan2013making}. As a result, we have the following remark.

\begin{remark}\ 
\begin{enumerate}
\item[$(1)$] $\sqcap{}'T_P=\sqcap{}T_P$.
\item[$(2)$] $\sqcap{}'T_Q=\sqcap{}T_Q$.
\end{enumerate}
\end{remark}

\section{$\sqcap'$-Tractable Queries without re-factorization}
\label{section_nofactorization}
As stated at the beginning of Section~\ref{section_reducibility}, this section is devoted to exploring the properties of $\sqcap'T_Q^0$.

Let us recall the NC-fcr reduction defined in Section~\ref{section_reducibility}.
Suppose that there are two classes $\mathcal{Q}_1$ and $\mathcal{Q}_2$ of Boolean queries, represented by languages $S_1$ and $S_2$ of pairs, respectively.
We say that $\mathcal{Q}_1$ is NC-fcr reducible to $\mathcal{Q}_2$ if there exist a cr-factorization $\Upsilon_1$ of the decision problem for $\mathcal{Q}_1$ and a cr-factorization $\Upsilon_2$ of the decision problem for $\mathcal{Q}_2$, such that we can transform
$S_{(\mathcal{Q}_1,\Upsilon_1)}$ to $S_{(\mathcal{Q}_2,\Upsilon_2)}$ using two NC functions
for data part and query part, respectively.
However, $S_{(\mathcal{Q}_1,\Upsilon_1)}$ may be different from $S_1$ and $S_{(\mathcal{Q}_2,\Upsilon_2)}$ may be different from $S_2$.
It is a natural idea to restrict transforms from the data of $S_1$ to that of $S_2$ and from the query of $S_1$ to that of $S_2$ without any re-factorization.
Based on this idea, Fan et al. defined F-reduction and studied the complete problem and size of $\sqcap{}T_Q^0$ \cite{fan2013making}.
They concluded that unless $\text{P}=\text{NC}$, $\sqcap{}T_Q^0$ is properly contained in P and a complete query class for $\sqcap{}T_Q^0$ is a witness in $\text{P}\backslash\text{NC}$.

With a little surprise, we find that our $\sqcap'T_Q^0$ is properly contained in P, regardless of whether $\text{P}=\text{NC}$.
Further, we will prove that the complete query of $\sqcap'T_Q^0$ is also a witness of $\text{P}\backslash\text{NC}$, and moreover, we get that $\sqcap'T_Q^0$ is properly contained in $\sqcap{}T_Q^0$ by showing that the query $\mathcal{Q}_\text{BDS}$ defined in Example~\ref{example_bdsquery} is not in $\sqcap'T_Q^0$.

We discuss the complete query class of $\sqcap'T_Q^0$ in Section~\ref{section_completeof0} and explain why $\mathcal{Q}_\text{BDS}$ defined in Example~\ref{example_bdsquery} is not in our $\sqcap'T_Q^0$ and show that $\sqcap'T_Q^0\subsetneq\sqcap{}T_Q^0$ in Section~\ref{section_difference}.

\subsection{Complete Queries for $\sqcap'T_Q^0$}
\label{section_completeof0}
For later need, we consider a reduction in $\sqcap'T_Q^0$, we would like to follow the concept of reduction $\leq_F^{NC}$ defined in \cite{fan2013making}.
Recall that as stated in Definition~\ref{definition_ncfreduction}, a language $S_1$ is said to be \textit{NC-F reducible} to another language $S_2$ of pairs, denoted by $S_1\leq^{NC}_FS_2$,
if there exist two NC functions $\alpha(\cdot),\beta(\cdot)$, such that for all $D,Q\in\Sigma^*$: $\langle{}D,Q\rangle\in{}S_1$
iff $\langle{}\alpha(D),\beta(Q)\rangle\in{}S_2$.
A class $\mathcal{Q}_1$ of Boolean queries is said to be \textit{NC-F reducible} to another class $\mathcal{Q}_2$ of Boolean queries,
denoted by $\mathcal{Q}_1\leq^{NC}_F\mathcal{Q}_2$, if $S_1\leq^{NC}_FS_2$, where $S_1$ and $S_2$
are the language pairs of $\mathcal{Q}_1$ and $\mathcal{Q}_2$, respectively.

Compared to Definition~\ref{definition_ncfcrreduction}, the reduction $\leq^{NC}_F$ is only defined on classes of Boolean queries, since decision problems do not have explicit data and query parts.
Moreover, this reduction does not allow re-factorization, so query classes are transformed strictly from data to data and from query to query.

As we will see in the following proposition, the reduction $\leq^{NC}_F$ is transitive and compatible in $\sqcap'T_Q^0$.
Note that, in contrast to $\leq^{NC}_{fcr}$, which is compatible in the set $\sqcap'T_Q$ of all queries that can be made $\sqcap'$-tractable, $\leq^{NC}_{F}$ is compatible in the set $\sqcap'T_Q^0$ of all queries that are $\sqcap'$-tractable.

\begin{proposition}
For any classes $\mathcal{Q}_1$, $\mathcal{Q}_2$, and $\mathcal{Q}_3$ of Boolean queries,
\begin{enumerate}
  \item[$(1)$] if $\mathcal{Q}_1\leq^{NC}_F\mathcal{Q}_2$ and $\mathcal{Q}_2\leq^{NC}_F\mathcal{Q}_3$,
            then $\mathcal{Q}_1\leq^{NC}_F\mathcal{Q}_3$.
  \item[$(2)$] If $\mathcal{Q}_1\leq^{NC}_F\mathcal{Q}_2$ and $\mathcal{Q}_2$ is in $\sqcap'T_Q^0$, then $\mathcal{Q}_1$ is also
            in $\sqcap'T_Q^0$. That is, a class of Boolean queries that is NC-F reducible to another $\sqcap'$-tractable class of Boolean queries is also $\sqcap'$-tractable.
\end{enumerate}
\end{proposition}
\begin{proof}
(1) Given $\mathcal{Q}_1\leq^{NC}_F\mathcal{Q}_2$ and $\mathcal{Q}_2\leq^{NC}_F\mathcal{Q}_3$, we have four NC functions $\alpha_1$, $\beta_1$, $\alpha_2$, and $\beta_2$, such that
\begin{equation*}
\begin{aligned}
  \langle{}D,Q\rangle\in{}S_1&\text{ iff }\langle\alpha_1(D),\beta_1(Q)\rangle\in{}S_2 \\
  \langle{}D,Q\rangle\in{}S_2&\text{ iff }\langle\alpha_2(D),\beta_2(Q)\rangle\in{}S_3
\end{aligned}
\end{equation*}
where $S_1$, $S_2$, and $S_3$ are the language of pairs for $\mathcal{Q}_1$,$\mathcal{Q}_2$, and $\mathcal{Q}_3$.

It is easy to see that
\begin{equation*}
  \langle{}D,Q\rangle\in{}S_1\text{ iff }\langle\alpha_2(\alpha_1(D)),\beta_2(\beta_1(Q))\rangle\in{}S_3.
\end{equation*}
Since $\alpha_2\comp\alpha_1$ and $\beta_2\comp\beta_1$ are also NC functions, it holds that $\mathcal{Q}_1\leq^{NC}_F\mathcal{Q}_3$.

(2) Given that $\mathcal{Q}_1\leq^{NC}_F\mathcal{Q}_2$, we have two NC functions $\alpha$ and $\beta$, such that
\begin{equation*}
  \langle{}D,Q\rangle\in{}S_1\text{ iff }\langle\alpha(D),\beta(Q)\rangle\in{}S_2
\end{equation*}
where $S_1$ and $S_2$ are the languages of pairs for $\mathcal{Q}_1$ and $\mathcal{Q}_2$.

Assuming that $\mathcal{Q}_2$ is $\sqcap'T_Q^0$, we have a PTIME function $\Pi$ and a language of pairs $S$ in NC, such that
\begin{equation*}
  \langle{}D,Q\rangle\in{}S_2\text{ iff }\langle\Pi(D),Q\rangle\in{}S.
\end{equation*}

Hence, we can verify that
\begin{equation*}
  \langle{}D,Q\rangle\in{}S_1\text{ iff }\langle\Pi(\alpha(D)),\beta(Q)\rangle\in{}S.
\end{equation*}
Let $S'$ be the set of pairs $\langle{}D,Q\rangle$ satisfying that $\langle{}D,\beta(Q)\rangle\in{}S$.
We can verify that $S'$ is also in NC and
\begin{equation*}
  \langle{}D,Q\rangle\in{}S_1\text{ iff }\langle\Pi(\alpha(D)),Q\rangle\in{}S'.
\end{equation*}

Because $\Pi$ is a PTIME function and has the short-output property, $\Pi\comp\alpha_2$ is also a PTIME function and has the short-output property.
Thus $S_1$ is $\sqcap'$-tractable and $\mathcal{Q}_1$ is in $\sqcap'T_Q^0$.
\end{proof}

After attaching the reduction $\leq_F^{NC}$ to $\sqcap'T_Q^0$, we study the complete queries for $\sqcap'T_Q^0$.
First, we define complete queries as follows.

\begin{definition}
  A class $\mathcal{Q}$ of Boolean queries is \textit{$\sqcap'T_Q^0$-hard} under NC-F reduction
  if $\mathcal{Q}'\leq_F^{NC}\mathcal{Q}$ for all $\mathcal{Q}'\in\sqcap'T_Q^0$.
  A class $\mathcal{Q}$ of Boolean queries is \textit{$\sqcap'T_Q^0$-complete} under NC-F reduction if $\mathcal{Q}$ is $\sqcap'T_Q^0$-hard and $\mathcal{Q}$ is $\sqcap'$-tractable itself.
\end{definition}

However, a complete query for $\sqcap'T_Q^0$ is also difficult to find and the existence of such a query
has a close relation with whether $\text{P}=\text{NC}$.
This property is the same as that of $\sqcap{}T_Q^0$.

\begin{theorem}
  A complete class of Boolean queries for $\sqcap'T_Q^0$ under NC-F reduction is a witness in $\text{P}\backslash\text{NC}$, unless $\text{P}=\text{NC}$.
\end{theorem}
\begin{proof}
  Suppose that $\text{P}\neq\text{NC}$ and there exists a class $\mathcal{Q}$ of Boolean queries, which is complete for $\sqcap'T_Q^0$
  under NC-F reduction.
  Clearly, $\mathcal{Q}$ is in P, so we only need to prove that $\mathcal{Q}\notin\text{NC}$.

  From Example~\ref{example_cvp}, we know that $\mathcal{Q}_\text{CVP}$ is in $\sqcap'T_Q^0$. It follows that $\mathcal{Q}_\text{CVP}\leq^{NC}_{F}\mathcal{Q}$, since $\mathcal{Q}$ is a complete query.
  That is, there exist two NC functions $\alpha(\cdot)$ and $\beta(\cdot)$, such that
  \begin{equation}\label{equation_ncnfreduction}
  \langle{}D,Q\rangle\in\mathcal{Q}_\text{CVP}\text{ iff }\langle\alpha(D),\beta(Q)\rangle\in\mathcal{Q}.
  \end{equation}
  It has been known that CVP is a P-complete problem \cite{greenlaw1995limits}, so for any decision problem $L\in\text{P}$,
  there exists an NC function $h$, such that $x\in{}L$ iff $h(x)\in\text{CVP}$.
  Note that $h(x)\in\text{CVP}$ iff $\langle{}h(x),\epsilon\rangle\in\mathcal{Q}_\text{CVP}$.
  Consequently it follows from Equation~(\ref{equation_ncnfreduction}) that $x\in{}L$ iff $\langle\alpha(h(x)),\beta(\epsilon)\rangle\in\mathcal{Q}$.
  So we can decide whether $x\in{}L$ by deciding whether $\langle\alpha(\phi(x)),\beta(\epsilon)\rangle\in\mathcal{Q}$.
  If $\mathcal{Q}\in\text{NC}$, then the latter can be decided in NC, and thus $L\in\text{NC}$. This forces that $\text{P}=\text{NC}$, a contradiction.
  Hence, $\mathcal{Q}$ is a witness in $\text{P}\backslash\text{NC}$.
\end{proof}

According to the above theorem, an approach to proving that $\text{P}\neq\text{NC}$ is to find a complete query for $\sqcap'T_Q^0$ and verify that it is not in NC.
Similarly, it is sufficient to prove that $\text{P}=\text{NC}$ by showing that a complete query for $\sqcap'T_Q^0$ is in NC.

\subsection{Difference between $\sqcap'T_Q^0$ and $\sqcap{}T_Q^0$}
\label{section_difference}
Until now, we have seen that $\sqcap'T_Q^0$ and $\sqcap{}T_Q^0$ share the same properties.
However, in this section, we show that $\sqcap'T_Q^0$ is properly contained in $\sqcap{}T_Q^0$ by verifying that a special query is not in our $\sqcap'T_Q^0$.
Our short-query restriction on the preprocessing function plays a role in the following theorem.

\begin{theorem}
  The class $\mathcal{Q}_\text{BDS}$ of Boolean queries is not $\sqcap'$-tractable, i.e., $\mathcal{Q}_\text{BDS}\notin\sqcap'T_Q^0$.
  Consequently, $\sqcap'T_Q^0\subsetneq\sqcap{}T_Q^0$ and $\sqcap'T_Q^0\subsetneq\text{P}$.
\end{theorem}
\begin{proof}
  By contradiction, suppose that $\mathcal{Q}_\text{BDS}\in\sqcap'T_Q^0$. Then there exist a preprocessing function $\Pi$ and a language $S'$
  of pairs which is in NC, such that,
  \begin{enumerate}
    \item $\langle{}G,(u,v)\rangle\in{}S_{\mathcal{Q}_\text{BDS}}$ iff $\langle\Pi(G),(u,v)\rangle\in{}S'$,
    \item $\Pi$ is a PTIME function and for each $x$, $|\Pi(x)|\in{}O(\text{polylog}(|x|))$.
  \end{enumerate}
  Supposing that $G$ has $n$ vertexes, we see that $|G|\in{}O(n^2)$. According to the polylog-size output property of $\Pi$, $|\Pi(G)|\in{}O(\text{polylog}(n^2))=O(\text{polylog}(n))$.
  Because $O(\text{polylog}(n))$ is a sub-linear size, we can get that $|\Pi(G)|<n$ when $n$ increases.

  On the other hand, given $S'$, we can treat $\Pi(G)$ as a function $f_{\Pi(G)}$ that takes a pair $(u,v)$ of nodes as input and outputs 1 if $\langle\Pi(G),(u,v)\rangle\in{}S'$ or 0, otherwise.
  That is, $f_{\Pi(G)}(u,v)=1$ iff $\langle\Pi(G),(u,v)\rangle\in{}S'$, equivalent to that $\langle{}G,(u,v)\rangle\in{}S_{\mathcal{Q}_\text{BDS}}$, which means $u$ appears before $v$ in the Breadth-Depth Search of $G$.
  Counting such functions, there are $n!$ functions in total, each for a permutation of all $n$ nodes of $G$, representing the list of nodes in the same order as they are visited during the search.

  Focusing on $|\Pi(G)|$, it is proven that $|\Pi(G)|<n$ when $n$ is large enough, so the count of $\Pi(G)$ is at most $2^n$ (the base is a constant number determined by the encoding method and we choose 2 here).
  Since $2^n<n!$ when $n$ is large enough, there must exist some function $f_{\Pi(G)}$ that cannot be derived from any $\Pi(G)$, which is a contradiction.
  Whence, $\mathcal{Q}_\text{BDS}$ is not $\sqcap'$-tractable.

  Note that in \cite{fan2013making}, it is proven that $\mathcal{Q}_\text{BDS}\in\sqcap{}T_Q^0$.
  This, together with the previous argument of $\mathcal{Q}_\text{BDS}\notin\sqcap'T_Q^0$, yields that $\sqcap'T_Q^0$ is properly contained in $\sqcap{}T_Q^0$.
  In addition, $\sqcap{}T_Q^0$ is contained in P \cite{fan2013making}, so our complexity class $\sqcap'T_Q^0$ is properly contained in P regardless of whether $\text{P}=\text{NC}$.
  This completes the proof of the theorem.
\end{proof}

\section{Conclusion}
\label{section_conclusion}
In the paper, we have pursued Fan et al.'s methodology \cite{fan2013making} of providing a formal foundation, in terms of computational complexity, for studying the tractability of query classes in the context of big data.
More specifically, motivated by $\sqcap$-tractability, we have introduced a notion of $\sqcap'$-tractable queries, which is more feasible on big data.
Such queries can be processed in parallel poly-logarithm time (NC) after a one-time PTIME data preprocessing, while the preprocessing function guarantees a polylog-size output.
We have addressed the decision problems and classes of Boolean queries that can be made $\sqcap'$-tractable,
based on our factorizations with constant redundancy. Moreover, we have discussed an NC reduction $\leq_{fcr}^{NC}$ and showed that it is transitive and compatible in the complexity class $\sqcap{}'T_P$ of decision problems and in the complexity class $\sqcap{}'T_Q$ of classes of Boolean queries that can be made $\sqcap'$-tractable.
In addition, we have shown that under the reduction $\leq_{fcr}^{NC}$, \mbox{BDS} and $\mathcal{Q}_\text{BDS}$ are $\sqcap{}'T_P$-complete and $\sqcap{}'T_Q$-complete, respectively.
For the complexity class $\sqcap'T_Q^0$, the set of all $\sqcap'$-tractable queries, we have discussed another existing NC reduction $\leq^{NC}_{F}$,
not allowing the re-factorization of data and query parts.
We have shown that under $\leq^{NC}_{F}$, the existence of a complete query for $\sqcap'T_Q^0$ is closely related to whether $\text{P}=\text{NC}$.
Moreover, we have proven that compared to $\sqcap{}T_Q^0$, our $\sqcap'T_Q^0$ is strictly smaller and thus $\sqcap'T_Q^0$ is properly contained in P.

\begin{figure}
\centerline{\includegraphics[width=7cm, trim=8cm 4.4cm 13cm 0.6cm, clip=true]{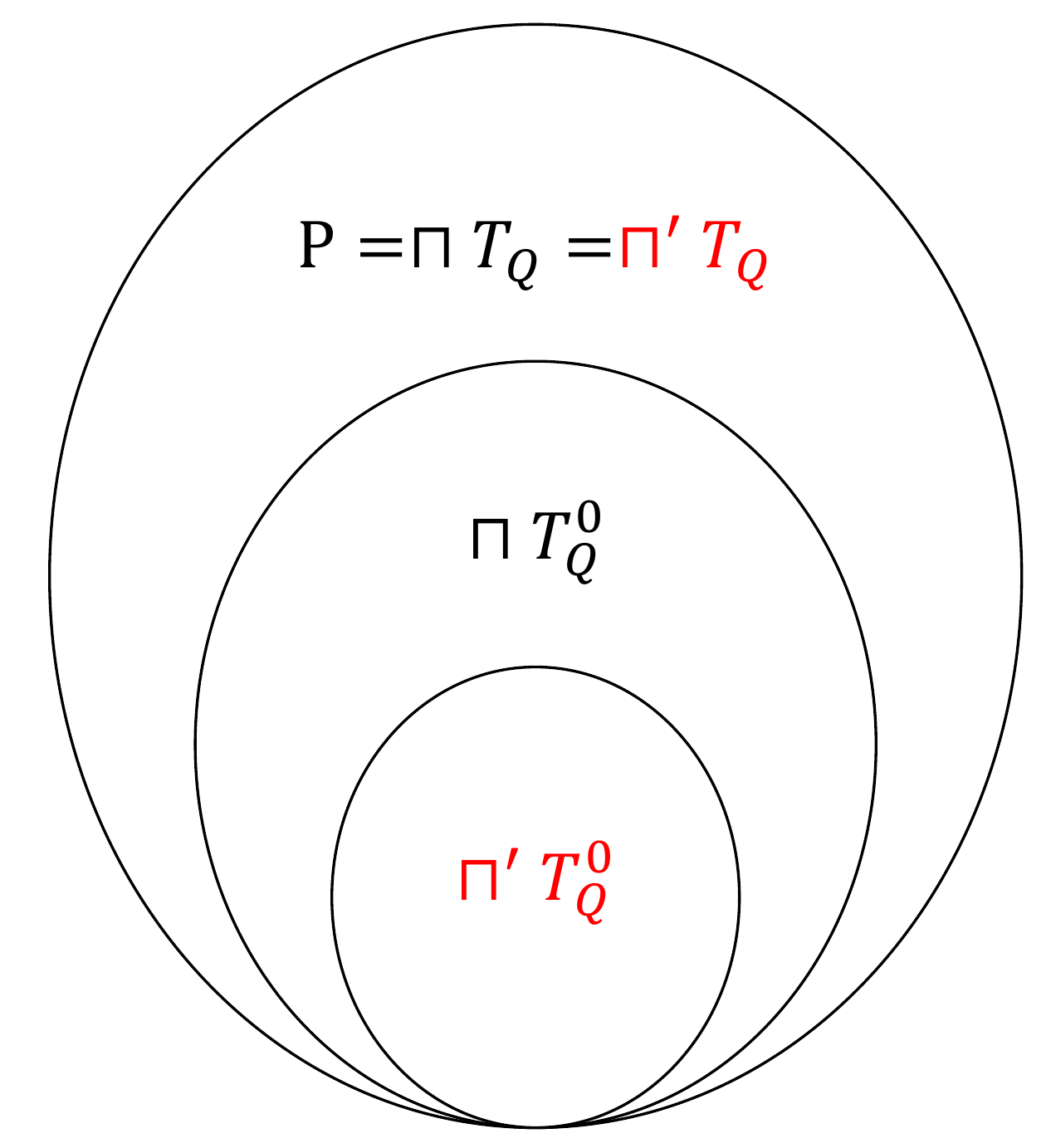}}
\caption{Venn diagram of related complexity classes, where $\sqcap'T_Q^0\subsetneq\sqcap{}T_Q^0$, and $\sqcap{}T_Q^0\subsetneq\text{P}$ unless $\text{P}=\text{NC}$.}
\label{figure_2}
\end{figure}

We sum up these main results as follows and display them in the Venn diagram of related complexity classes in Figure~\ref{figure_2}:
\begin{enumerate}
\item $\sqcap'T_Q^0\subsetneq\sqcap{}T_Q^0$, and thus $\sqcap'T_Q^0\subsetneq\text{P}$, regardless of whether $\text{P}=\text{NC}$.
\item All classes of Boolean queries in P can be made $\sqcap'$-tractable.
\item For decision problems, $\sqcap'T_P=\text{P}$.
\end{enumerate}

There are two problems which are worth further studying.
First, in the present work we have focused on the preprocessing functions with short-output restriction.
In fact, there are many other characteristics of the preprocessing functions utilized to deal with big data, such as dividing a big graph into clusters for the subsequent preprocessing \cite{Yang20141563,kang2011gbase}.
It is interesting to develop the corresponding tractability of query classes according to different characteristics.
Second, note that NC, which is a measure on time, has been considered as a feasible solution after preprocessing.
Therefore, it would be interesting to take into account the measures on the other aspects, such as space, of feasible solutions.

\section*{Acknowledgements}
The authors would like to thank Professor Wenfei Fan for his invaluable suggestions.
This work was supported by the National Natural Science Foundation of China (Grants No. 61170299 and 61370053).


\end{document}